\documentclass[a4paper,11pt]{article}

\usepackage{amssymb}
\usepackage{microtype}
\usepackage{url}
\usepackage{amsmath}
\usepackage{courier}
\usepackage{graphicx}
\usepackage{multirow}
\usepackage{verbatim}
\usepackage{color}

\usepackage{xcolor}
\usepackage[colorlinks = true,
            linkcolor = red,
            urlcolor  = blue,
            citecolor = blue,
            anchorcolor = blue]{hyperref}

\usepackage{tikz}
\usetikzlibrary{arrows}

\newcommand{\BIGCOMMENT}[1]{}

\newcommand{\COMMENT}[1]{}

\newcommand{\ket}[1]{|#1\rangle}

\usepackage{amsthm}

\newtheorem{definition}{Definition}

\newtheorem*{theorem*}{Theorem}

\newtheorem{theorem}[definition]{Theorem}

\newtheorem{claim}{Claim}

\newtheorem{fact}{Fact}

\newtheorem{question}{Question}

\newtheorem{corollary}{Corollary}

\newtheorem{lemma}{Lemma}

\newtheorem{remark}{Remark}

\title{Quantum Query Complexity of Subgraph Isomorphism and Homomorphism}

\author{Raghav Kulkarni\\ Centre for Quantum Technologies \& Nanyang Technological University,\\ Singapore \\ \texttt{kulraghav@gmail.com}
\and \\ Supartha Podder\\  Centre for Quantum Technologies, National University of Singapore,\\ Singapore\\ \texttt{supartha@gmail.com}}

\date{}

\begin{document}

\maketitle

\begin{abstract}
Let $H$ be a fixed graph on $n$ vertices. Let $f_H(G) = 1$ iff
the input graph $G$ on $n$ vertices contains $H$ as a (not necessarily induced) subgraph. 
Let $\alpha_H$ denote the cardinality of a maximum independent set of $H$.
In this paper we show:
\[Q(f_H) = \Omega\left( \sqrt{\alpha_H \cdot n}\right),\] 
where $Q(f_H)$ denotes the quantum query complexity of $f_H$.

As a consequence we obtain a lower bounds for $Q(f_H)$ 
in terms of several other parameters of $H$ such as
the average degree, minimum vertex cover, chromatic number, and the critical probability.

We also use the above bound to show that $Q(f_H) = \Omega(n^{3/4})$ for any $H$, 
improving on the
previously best known bound of $\Omega(n^{2/3})$ \cite{sy}. 
Until very recently, it was believed that the quantum query complexity is at least square root of the randomized one.
Our $\Omega(n^{3/4})$ bound for $Q(f_H)$ matches the square root of the current best known bound
for the randomized query complexity of $f_H$, which is $\Omega(n^{3/2})$
due to Gr\"oger \cite{groger}.
Interestingly, the randomized bound of $\Omega(\alpha_H \cdot n)$ for $f_H$ still remains open.

We also study the Subgraph Homomorphism Problem, denoted by $f_{[H]}$, and show that $Q(f_{[H]}) = \Omega(n)$.

Finally we extend our results to the $3$-uniform hypergraphs. In particular, we show an $\Omega(n^{4/5})$ bound for quantum query complexity of the Subgraph Isomorphism, improving on the previously known $\Omega(n^{3/4})$ bound.
For the Subgraph Homomorphism, we obtain an $\Omega(n^{3/2})$ bound for the same.

\end{abstract}

{\bf keywords:} Quantum Query Complexity, Subgraph Isomorphism, Monotone Graph Properties.

\section{Introduction}
\label{sec:intro}
\subsection{Classical and Quantum Query Complexity}
The decision tree model (aka the query model), perhaps due to its  simplicity and fundamental nature, has been extensively studied in the past
and still remains a rich source of many fascinating investigations.
In this paper we focus on Boolean functions, i.e., the functions of the form $f : \{0,1\}^n \to \{0, 1\}.$
 A deterministic decision tree $T_f$ for $f$
takes $x = (x_1, \ldots, x_n)$ as an input and determines the value
of $f(x_1, \ldots, x_n)$ using queries of the form 
$\text{``is } x_i = 1? \text{"}$  Let $C (T_f, x)$ denote the cost of the computation, that is the number of queries made by $T_f$ on
an input $x.$ The {\em deterministic decision tree complexity} (aka the deterministic query complexity) of $f$ is defined as
\[D(f) = \mathop{\min}_{T_f} \max_x C (T_f, x). \]

We encourage the reader to see an excellent survey by Buhrman and de~Wolf \cite{bw} 
on the decision tree complexity of Boolean functions.

A randomized decision tree $\mathcal{T}$ is simply a probability distribution on the deterministic decision trees $\{T_1, T_2, \ldots\}$ where the tree $T_i$ occurs with
probability $p_i.$ We say that $\mathcal{T}$ computes $f$ correctly if for every input
$x$: $\Pr_{i}[T_i(x) = f(x)] \geq 2/3$. The depth of $\mathcal{T}$ is the maximum depth
of a $T_i$. {\em The (bounded-error) randomized query complexity} of $f$, denoted by $R(f)$,  is the minimum possible
depth of a randomized tree computing $f$ correctly on all inputs. 

One can also define the quantum version of the decision tree model as follows:
Start with an $N$-qubit state $\ket{0}$ consisting of all zeros. We can transform this state by applying an unitary transformation $U_0$, then we can make a {\em quantum} query $O$, which essentially negates the amplitude of each basic state depending on whether the $i$th bit of the basic state is zero or one. A quantum algorithm with $q$ queries looks like the following: $A = U_q O U_{q-1} \cdots O U_1 O U_0$. Here $U_i$'s are fixed unitary transformation independent of the input $x$. The final state $A\ket{0}$ depends on the input $x$ only via applications of $O$. We measure the final state outputing the rightmost qubit (WLOG there are no intermediate measurements). 
A bounded-error quantum query algorithm $A$ computes $f$ correctly if the final measurement gives the correct answer with probability at least $2/3$ for every input $x$.
{\em The bounded-error quantum query complexity} of $f$, denoted by $Q(f)$, is the least $q$ for which $f$ admits a bounded-error quantum algorithm. 
We refer the reader to a survey by Buhrman and de~Wolf \cite{bw} for more precise definition.

\subsection{Subgraph Isomorphism Problem}
Let $H$ be a fixed graph on $n$ vertices (possibly with isolated vertices) and let $G$ be an unknown input graph (on $n$ vertices) given by query access to its edges, i.e, queries of the form
``Is $\{i,j\}$ an edge in $G$?". 
We say $H \leq G$ if $G$ contains $H$ as a (not necessarily induced) subgraph.
 Let $f_H : \{0,1\}^{n \choose 2} \to \{0, 1\}$ be
defined as follows: 
\begin{equation}
f_H(G) =
\left\{
	\begin{array}{ll}
		1  & \mbox{if  }  H \leq G\\
		0 & \mbox{} \text{otherwise}
	\end{array}
\right.
\end{equation}

The well-known Graph Isomorphism Problem asks whether a graph $H$ is isomorphic to another graph $G$. The 
Subgraph Isomorphism Problem is a generalization of the Graph Isomorphism Problem where one asks whether $H$ is isomorphic to a subgraph of $G$.
Several central computational problems for graphs such as containing a clique, containing a Hamiltonian cycle, containing a perfect matching can be formulated as the Subgraph Isomorphism Problem by fixing the $H$ appropriately.  Given the generality
and importance of the problem, people have investigated various restricted special cases of this problem in different models of computation \cite{wiki} \cite{lrr}. In the context of query complexity, 
in 1992 Gr\"oger \cite{groger} studied this problem in the randomized setting and showed that $R(f_H) = \Omega(n^{3/2})$, which is the best known bound to this date. In this paper we investigate this problem in the quantum setting. To the best of our knowledge, quantum query complexity for the Subgraph Isomorphism Problem has not been noted prior to this work when $H$ is allowed to be any graph on $n$ vertices. A special case of this problem when $H$ is of a constant size has been investigated before for obtaining upper bounds~\cite{LMS}.

\subsection{Subgraph Homomorphism Problem}
 We also investigate a closely related Subgraph Homomorphism Problem.

 A homomorphism from a graph $H$ into a graph $G$ is a function $h : V(H) \to V(G)$
such that: if $(u,v) \in E(H)$ then $(h(u), h(v)) \in E(G)$.

Let $f_{[H]}$ be the function defined as follows:
$f_{[H]}(G) = 1$ if and only if $H$ admits a homomorphism into $G.$

Note that unlike the isomorphism, the homomorphism need not be an injective function from $V(H)$ to $V(G)$.
We study the query complexity of the Subgraph Homomorphism Problem towards the end of this paper.
In the next section, we review the relevant literature.

\subsection{Related Work}
Understanding the query complexity of monotone graph properties has a long history. In the deterministic setting the Aanderaa-Rosenberg-Karp Conjecture asserts that one must query all the $n \choose 2$ edges in the worst-case. The randomized complexity of monotone graph properties is conjectured to be $\Omega(n^2)$. Yao \cite{yao87} obtained the first
super-linear lower bound in the randomized setting using the graph packing arguments. 
Subsequently his bound was improved by King \cite{king} and later by Hajnal \cite{hajnal}. The current best known bound is
$\Omega(n^{4/3}\sqrt{\log n})$ due to Chakrabarti and Khot \cite{ck}. Moreover, O'Donnell, Saks, Schramm, and Servedio \cite{osss} also obtained an $\Omega(n^{4/3})$ bound via a more generic approach for monotone transitive functions. 
Friedgut, Kahn, and Wigderson \cite{fkw} obtain an $\Omega(n/p)$ bound where the $p$ is the critical probability 
of the property. 
In the quantum setting, Buhrman, Cleve, de~Wolf and Zalka \cite{bcwz} were the first to study quantum complexity of graph properties. Santha and Yao \cite{sy} obtain an $\Omega(n^{2/3})$ bound for general properties. Their proof follows 
along the lines of Hajnal's proof. 

Gr\"oger \cite{groger} obtained
an $\Omega(n^{3/2})$ bound for the randomized query complexity of the Subgraph Isomorphism.
This is currently the best known bound for the Subgraph Isomorphism Problem. 
Until very recently\footnote{Very recently this has been falsified by Ben-David \cite{bd}.}, it was believed that the quantum query complexity is at least square root of the randomized one. In this paper we address the quantum query complexity of the Subgraph Isomorphism Problem and obtain the square root of the current best randomized bound.

The main difference between the previous work and this one is that all
the previous work, including that of Santha and Yao \cite{sy}, obtained the lower bounds based on an embedding of a {\em tribe} function \cite{bbcmw} on a large number of variables
 in monotone graph properties. Recall that the tribe function with parameters $k$ and $\ell$, is a function  $T(k, \ell)$ on $k \cdot \ell$ variables defined as: $\bigvee_{i \in [k]} \bigwedge_{j \in [\ell]} x_{ij}$.
This method yields a lower bound of $\Omega(k \cdot \ell)$ for the randomized query complexity and $\Omega(\sqrt {k \cdot \ell})$ for the quantum.
 We deviate from this line by embedding a threshold function $T_n^t$ instead of a tribe.
 Recall that $T_n^t(z_1, \ldots, z_n)$ is a function on $n$ variables that evaluates to $1$ if and only if at least $t$ of the $z_i$'s are $1$.
 Since the randomized complexity of $T_n^t$ is $\Theta(n)$, this does not give us any advantage for obtaining super-linear randomized
 lower bounds. However, it {\em does} yield an advantage for the quantum lower bounds as the quantum
 query complexity of $T_n^t$ is $\Theta(\sqrt{n(n-t)})$, which can reach up to $\Omega(n)$ for large $t$. Since this technique works only in the quantum setting, the randomized versions of our bounds remain intriguingly open.

\subsection{Our Results}
Our main result is a lower bound on the quantum query complexity of the Subgraph Isomorphism Problem for $H$ in terms of the maximum independence number of $H$.
\begin{theorem} For any $H$,
\[ Q(f_H) = \Omega\left(\sqrt{\alpha_H \cdot n} \right),\]
where $\alpha_H$ denotes the size of a maximum independent set of $H$.
\label{thm:alpha}
\end{theorem}
\begin{corollary} For any $H$,
\begin{enumerate}
\item $Q(f_H) = \Omega\left(\frac{n}{\sqrt{d_{avg}(H)}}\right),$
\item $Q(f_H) = \Omega\left(\frac{n}{\sqrt {\chi_H}}\right),$
\item $ Q(f_H) = \Omega\left(\sqrt{\frac{n}{p}}\right),$
\end{enumerate}
where $d_{avg}(H)$ denotes the average degree of the vertices of $H$, $\chi_H$ denotes the chromatic number of $H$, and
$p$ denotes the critical probability \cite{fkw} of $H$.
\label{cor:d-avg}
\label{cor:chi}
\label{cor:n/p}
\end{corollary}
In particular, we get an $\Omega(n)$ bound when the graph $H$ is sparse ($|E(H)| = O(n)$),
or $H$ has a constant chromatic number, or the critical probability of $H$ is $O(1/n)$. 
Friedgut, Kahn, and Wigderson \cite{fkw} show an $\Omega(n/p)$ bound for the randomized query complexity of general monotone properties. Quantization of this bound remains open.
General monotone properties can be thought of as the Subgraph Isomorphism for a {\em family}
of minimal subgraphs. The item 3 above, gives a quantization of \cite{fkw}
in the case when the family contains only a single subgraph.  

\begin{corollary}
For any $H$,
\[Q(f_H) = \Omega(n^{3/4}). \]
\label{cor:iso}.
\end{corollary}
Prior to this work only an $\Omega(n^{2/3})$ bound was known from the work of Santha and Yao \cite{sy} on general monotone graph properties.

We extend our result to the $3$-uniform hypergraphs. In particular, we show:
\begin{theorem}
\label{thm:4over5}
\label{thm:3-uniform}
 Let $H$ be a $3$-uniform hypergraph on $n$ vertices. Then,
\[Q(f_H) = \Omega(n^{4/5}).\]
\end{theorem}
This improves the $\Omega(n^{3/4})$ bound obtained via the minimum certificate size.

The second part of this paper concerns the Subgraph Homomorphism Problem for $H$, denoted by $f_{[H]}$. 
Here we show the following two Theorems: 
\begin{theorem}
For any $H$,
\[Q(f_{[H]}) = \Omega(n).\]
\label{thm:homo}
\label{thm-homo}
\end{theorem}

\begin{theorem}
\label{thm:homo2}
\label{thm:homo-3uniform}
For any $3$-uniform hypergraph $H$ on $n$ vertices:
\[Q(f_{[H]}) = \Omega(n^{3/2}).\]
\end{theorem}

Our proofs crucially rely on the duality of monotone functions and  appropriate embeddings of tribe and threshold functions.
 All our lower bounds hold for the approximate degree $\widetilde{\deg}(f)$, which is known to be strictly smaller than the quantum
query complexity \cite{amb}.

\subsection*{Organization}
Section~\ref{sec:prelims} contains some preliminaries. 
Section~\ref{sec:d-avg} and Section~\ref{sec:3-uniform} deal with the Subgraph Isomorphism Problem.
Section~\ref{sec:d-avg} contains the proofs of Theorem~\ref{thm:alpha}, Corollary~\ref{cor:d-avg} and Corollary~\ref{cor:iso}. 
Then Section~\ref{sec:3-uniform} contains the proof of Theorem~\ref{thm:4over5}. 
The next two sections (Section~\ref{sec:homo} and Section~\ref{sec:homo-appx}) involve the Subgraph Homomorphism Problem and contains the proof of Theorem~\ref{thm:homo} and Theorem~\ref{thm:homo2}.
Finally Section~\ref{sec:open} contains conclusion and some open ends.

\section{Preliminaries}
\label{sec:prelims}
In this section, we review some preliminary concepts and restate some previously known results.

Let $[n]$ denote the set $\{1, \ldots, n\}$.
\begin{definition}[Dual of a Property]
The dual $\mathcal{P}$, denoted by $\mathcal{P}^*$, is:
\[\mathcal{P}^*(x) := \neg \mathcal{P}(\neg x), \]
where $\neg x$ denotes the binary string obtained by flipping each bit in $x$. 

Note that $\mathcal{P}^{**} = \mathcal{P}$ and $Q(\mathcal{P}) = Q(\mathcal{P}^*)$.
\end{definition}

A property $\mathcal{P}$ is said to be {\em monotone increasing} if 
for every $x \leq y$ we have $\mathcal{P}(x) \leq \mathcal{P}(y)$, where $x \leq y$ denotes
$x_i \leq y_i$ for all $i$. 

Note that if $\mathcal{P}$ is monotone, then so is $\mathcal{P}^*$.

A {\em minimal certificate} of size $s$ for a monotone increasing property $\mathcal{P}$ is an
input $z$ such that (a) The hamming weight of $z$, i.e,  $|z|$, is $s$,
(b) $\mathcal{P}(z) = 1$, and (c)
for any $y$ with $|y| < s$, $\mathcal{P}(y) = 0$. Every minimal certificate $z$ can be uniquely associated with the subset $S_z := \{i \mid  z_i = 1\}$.
\begin{lemma}[Minimal Certificate \cite{bw}]
If $\mathcal{P}$ has a minimal certificate of size $s$ then
\[Q(\mathcal{P}) \geq \Omega(\sqrt s). \] 
\label{lem:certificate}
\end{lemma}
We say that two minimal certificates $z_1$ and $z_2$ {\em pack} together, if \\$S_{z_1} \cap S_{z_2} = \emptyset$.

\begin{lemma}[Packing Lemma \cite{yao87}] 
If $z_1$ is a minimal certificate of $\mathcal{P}$ and $z_2$ is a minimal certificate
of $\mathcal{P}^ *$ then $z_1$ and $z_2$ {\em cannot} be packed together.
\label{lem:packing}
\end{lemma}

\begin{lemma}[Tur\'an \cite{apss}]
\label{lem:turan}
If the average degree of a graph $G$ is $d$ then $G$ contains an independent 
set of size at least $\Omega(n/d)$.
\end{lemma}

\begin{lemma}[Extended Tur\'an \cite{apss}]
\label{lem:extended-turan}
If the average degree of a $k$-uniform hypergraph $G$ is $d$ then $G$ contains an independent 
set of size at least $\Omega(n/d^{\frac{1}{k-1}})$.
\end{lemma}

A Boolean function $f(x_1, \ldots, x_n)$ is called {\em transitive} if there exists
a group $\Gamma$ acting transitively on the $i$  such that $f$ is invariant under
the action, i.e., for every $\sigma \in \Gamma$ we have $f(x_{\sigma_1}, \ldots, x_{\sigma_n})
= f(x_1, \ldots, x_n)$.

Note that graph properties and hypergraph properties are transitive functions.
\begin{lemma}[Transitive Packing \cite{syz04}]
Let $f$ be a monotone transitive function on $n$ variables. If $f$ has a minimal certificate
of size $s$ then every certificate of $f^*$ must have size at least
$n/s$.
\label{lem:transitive-packing}
\end{lemma}

A {\em Threshold function} $T_n^t(z_1, \ldots, z_n)$ is a function on $n$ variables such that $T_n^t$ outputs
$1$ if and only if at least $t$ variables are $1$.

We are now ready to prove the quantum query complexity lower bound for the Subgraph Isomorphism Problem.

\section{Subgraph Isomorphism for Graphs}
\label{sec:d-avg}

Before proving Theorem~\ref{thm:alpha} we first prove two lemmas.

Let $S_d$ denote the star graph with $d$ edges. Then $f_{S_d}$ is the property of having a vertex of at least degree $d$. First we show:
\begin{lemma}
 \[ Q(f_{S_d}) = \Omega(n)\]
 \label{lem:star}
\end{lemma}
\begin{proof}
We divide the proof into two cases:

{\sf Case 1:} $d > n/2$.

Fix a clique on the vertices $1, \ldots, \lfloor n/2 \rfloor$ and fix an independent set on the vertices $\lfloor n/2 \rfloor + 1, \ldots, n$. Note that we still have $\lfloor n/2 \rfloor \times \lceil n/2 \rceil$ edge-variables that are not yet fixed.
Now as soon as any vertex $v$ from the clique has $(d -\lfloor n/2\rfloor +1)$ edges to the independent set present, we have a $d$-star. Thus $f_{S_d}$ becomes an $OR_{\lfloor n/2 \rfloor} \circ T^{(d-\lfloor\frac{n}{2}\rfloor+1)}_{\lceil n/2 \rceil} $ function, which has a lower bound of $\Omega(n)$ via the Composition Theorem for quantum query complexity \cite{LMR}. 

{\sf Case 2:} $d \leq n/2$.

A minimum certificate of $f_{S_d}$ is a $d$-star. Now by the Lemma~\ref{lem:packing} we know that this $d$-star can not be packed with any minimal certificate of the dual $f^*_{S_d}$. Thus every vertex in the dual $f^*_{S_d}$ must have degree $ > n - d$. 
Hence the minimal certificate size is at least $\Omega(n^2)$ and $Q(f^*_{S_d}) = Q(f_{S_d}) = \Omega(n)$.

\end{proof}

Let $t$ denote the smallest integer such that $f_H^*(K_t) = 1$. 
\begin{lemma}
\[ Q(f_H) \geq \Omega(\sqrt{n(n-t)}). \]
\label{lem:threshold}
\end{lemma}
\begin{proof}

We {\em embed} $T_n^t$ in $f_H^*$ (on inputs of Hamming weight $t-1$ and $t$) via the following mapping:
Let $x_{ij} := z_i \cdot z_j$.  Let $f'(z_1, \ldots , z_n) := f_H^*(\{x_{ij}\})$. 
Note that $f' \equiv T_n^t$. Also note\footnote{Since $x_{ij} = z_i.z_j$, every monomial of $f_H^*$ of size $d$ becomes a monomial of size at most $2d$ in $f'$.} that $Q(f_H) = Q(f_H^*)$ and $\widetilde{deg}(f') \leq 2\cdot\widetilde{deg}(f_H^*)$.
Since $Q(f) \geq \widetilde{deg}(f)$, it remains to prove the following:

\begin{claim}
$\widetilde{deg}(f') = \Omega(\sqrt{n(n-t)})$
\label{claim-n-n-t}
\end{claim}

We need the following lemma due to Paturi \cite{paturi}:

\begin{lemma}
\label{lem:paturi}
Let $g$ be a function on $n$ variables such that $g(z) = 0$ for all $z$ with $|z| = t-1$ and $g(z) = 1$ for all $z$ with $|z| = t$. Then: 
$\widetilde{\deg}(g) = \Omega(\sqrt{n(n-t)})$.
\end{lemma}

 \noindent
\textit{Proof of Claim \ref{claim-n-n-t}.}
Note that $f'$ ($\equiv T_n^t$) satisfies the condition of the Lemma~\ref{lem:paturi}.

\hfill $\Box$

This finishes the proof of the Lemma~\ref{lem:threshold}.

\end{proof}

Now we are in a position to prove the Theorem~\ref{thm:alpha}.

\subsection*{Proof of Theorem~\ref{thm:alpha}.}
Recall that $t$ denotes the smallest integer such that $f_H^*(K_t) = 1$. 
We divide the proof into two cases:

{\sf Case 1:} $t > n/2$

In this case, we reduce the $f_H$ to $f_{S_p}$ for some $p = \Omega(n)$. 
Let $\nu_H$ denote the minimum vertex cover size of $H$.
Since $t > n/2$, we have $\nu_H < n/2$. We restrict $f_H$ by picking a clique on $\nu_H - 1$ vertices
and joining all the other $n - \nu_H + 1$ remaining vertices to each vertex in this clique. The resulting function takes a graphs on $p = n-\nu_H + 1$ vertices as input. Let's denote these vertices by $S$.

As the clique on $\nu_H - 1$ vertices can not accommodate all the vertices in the minimum vertex cover, at least one vertex $v$ must occur among $S$.

In fact, it is not difficult to see that the property is now reduced to finding a star graph with $d$ edges, $f_{S_d}$ where $d$ is defined as follows:
Let $C$ be a vertex cover. Furthermore let $d_{out}(v)$ denote the number of neighbors of a vertex $v$ in $C$ that are outside $C$ and $d_{out}(C)$ be the minimum over all such vertices $v$ in $C$.
 Then $d$ is the minimum $d_{out}(C)$ 
 of a minimum vertex cover $C$ of $H$ (minimized over all the minimum vertex covers).

Now from the Lemma~\ref{lem:star} we get $Q(f_H) = \Omega(n)$.

{\sf Case 2:} $t \leq n/2$

From the Lemma~\ref{lem:packing} we conclude that $t > \alpha_H$. Since $t \leq n/2$, we have $n-t = \Omega(n)$.
Hence from the Lemma~\ref{lem:threshold} we get the bound of $\Omega(\sqrt{n(n-t)})$,
which is $\Omega(\sqrt{\alpha_H \cdot n})$.

\hfill $\Box$

\subsection*{Proof of Corollary~\ref{cor:d-avg}}
(1) 
From
Tur\'an's theorem, we have: $\alpha_H \geq n/(2 \cdot d_{avg}(H))$. 

\noindent (2) Since $\alpha_H \cdot \chi_H \geq n$ we have $\alpha_H \geq n/\chi_H$.

\noindent (3) Since the critical probability of $H$ is $p$, the average degree
of $H$ is at most $p n$. Hence from Corollary~\ref{cor:d-avg}(1), we get the $\Omega({n/p})$ bound.
\hfill $\Box$

\subsection*{Proof of Corollary~\ref{cor:iso}.}
When $d_{avg}(H) \geq \sqrt n$
the Lemma~\ref{lem:certificate} gives an $\Omega(n^{3/4})$ bound. Otherwise when $d_{avg}(H) < \sqrt n$ we use the Corollary~\ref{cor:d-avg}(1), which gives the same bound.

\hfill $\Box$

\section{Subgraph Isomorphism for $3$-Uniform Hypergraphs}
\label{sec:3-uniform}
In this section we extend 
the $\Omega(n^{3/4})$ bound for the Subgraph Isomorphism for graphs 
to the $3$-uniform hypergraphs. 
In particular, we obtain an $\Omega(n^{4/5})$ bound for the Subgraph Isomorphism for 
$3$-uniform hypergraphs, improving upon the $\Omega(n^{3/4})$ bound obtained via the minimal certificate size.

Before going to the proof of Theorem~\ref{thm:3-uniform}, we extend the Lemma~\ref{lem:threshold} to the $3$-uniform hypergraphs.
Let $t$ be the smallest such that $f_H^*(K_t) = 1$. Note that
$t = \alpha_H+1$.
\begin{lemma} Let $H$ be a $3$-uniform hypergraph on $n$ vertices. Then:
\[ Q(f_H) \geq \Omega(\sqrt{n(n-t)}). \]
\label{lem:3-uniform-threshold}
\end{lemma}
\begin{proof}
Let $T_n^t(z_1, \ldots, z_n)$ denote the threshold function on $n$ variables that outputs
$1$ if and only if at least $t$ variables are $1$.
We {\em embed} a $T_n^t$ in $f_H^*$ (on inputs of Hamming weight $t-1$ and $t$) via the following mapping:
Let $x_{ijk} := z_i \cdot z_j \cdot z_k$.  Let $f'(z_1, \ldots , z_n) := f_H^*(\{x_{ijk}\})$. 
Note that $f' \equiv T_n^t$. Also note\footnote{Since 
$x_{ijk} = z_i \cdot z_j \cdot z_k$, every monomial of $f_H^*$ of size $d$ becomes a monomial of size at most $3d$ in $f'$.} that the $\widetilde{deg}(f') \leq 3\cdot\widetilde{deg}(f_H^*)$.
Since $Q(f) \geq \widetilde{deg}(f)$, it remains to prove that $\widetilde{deg}(f') = \Omega(\sqrt{n (n-t)})$, which follows from the Lemma~\ref{lem:paturi}.

\end{proof}
Now we give a proof of the Theorem~\ref{thm:3-uniform}.

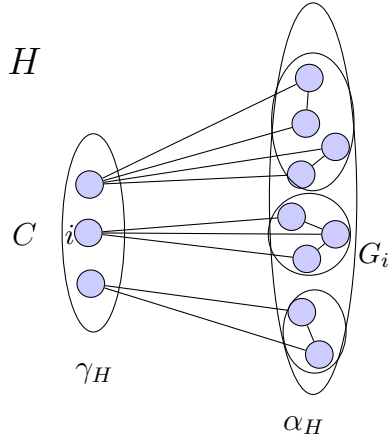
\begin{figure}[h]
\centering

\begin{tikzpicture}[>=stealth',shorten >=0pt,auto,node distance=1cm,
  thin,main node/.style={circle,fill=blue!20,draw,font=\sffamily\tiny\bfseries},scale=0.46]

\draw  (-2.7429,1.0286) ellipse (1.25 and 5.6429);
\draw  (-9.0289,0.0424) ellipse (0.8929 and 2.8571);
\node[main node] (1) at (-9.1369,1.4353) {};
\node[main node] (2) at (-9.1727,0.0424) {};
\node[main node] (3) at (-9.1012,-1.4218) {};

\node[main node] (4) at (-2.8512,4.5067) {};
\node[main node] (5) at (-2.9512,3.1924) {};

\node[main node] (11) at (-2.1012,2.521) {};
\node[main node] (12) at (-3.084,1.721) {};

\node[main node] (6) at (-3.3584,0.5067) {};
\node[main node] (7) at (-2.1012,0) {};
\node[main node] (8) at (-2.9227,-0.7076) {};

\node[main node] (9) at (-3.0655,-2.2433) {};
\node[main node] (10) at (-2.5655,-3.4647) {};

\draw  (-2.7512,3.2424) ellipse (1.1786 and 1.9857);
\draw  (-2.85,0) ellipse (1.175 and 1.18);
\draw  (-2.7,-2.8) ellipse (0.9 and 1.2);

\draw (1) edge (5);
\draw (1) edge (4);
\draw (4) edge (5);

\draw (1) edge (11);
\draw (1) edge (12);
\draw (11) edge (12);

\draw (2) edge (6);
\draw (2) edge (7);
\draw (2) edge (8);
\draw (6) edge (7);
\draw (7) edge (8);

\draw (3) edge (9);
\draw (3) edge (10);
\draw (9) edge (10);

      \node[style={font=\sffamily\large\bfseries}] at (-9,-4) {$\gamma_H$};
      
      \node[style={font=\sffamily\large\bfseries}] at (-3,-5.5) {$\alpha_H$};

      \node[style={font=\sffamily\Large\bfseries}] at (-11,5) {$H$};
      \node[style={font=\sffamily\large\bfseries}] at (-11,0) {$C$};

      \node[style={font=\sffamily\bfseries}] at (-9.7,0) {$i$};
      \node[style={font=\sffamily\bfseries}] at (-1.0,-0.6) {$G_i$};

\end{tikzpicture}
\caption{Structure of $H$}\label{fig:H}
\end{figure}

\subsection*{Proof of Theorem~\ref{thm:3-uniform}.}
We divide the proof into two main cases.

\noindent {\sf Case 1:} $\alpha_H > n/2$.
Let $H$ be a 3-uniform hypergraph on $n$ vertices. Let $C$ denote a minimal vertex cover of $H$. Let $|C|= \nu_H$. Note that the hypergraph induced on $V-C$ is empty. For a vertex $i \in C$ let $G_i$ denote the projection graph of the neighbors of $i$ on $V-C$, i.e., $(i,u,v) \in E(H)$ (See Figure~\ref{fig:H}). 

 Let $\mathcal{P}_H$ denote the restriction of the $f_H$ defined as follows:
set the hyper-clique on $\nu_H - 1$ vertices to be present and add all the hyper-edges incident on the vertices of this clique. 
Let $S$ denote the set of remaining $n - \nu_H + 1$ vertices.
The hyper-edges among $S$ are still undetermined. Note that $\mathcal{P}_H$ is a non-trivial property of $n-\nu_H +1$ vertex hypergraphs, since $H$ cannot be contained in the $\nu_H - 1$ hyper-clique and edges incident on it as the minimum vertex cover size of $H$ is $\nu_H$.

\begin{lemma}
\label{lem:pivot}
If $\exists C$,  $\exists i : |E(G_i)| = O(n^{7/5})$, then $Q(f_H) =\Omega(n^{4/5})$.
\end{lemma}
\begin{proof}
In this case $\mathcal{P}_H$ has a certificate of size $O(n^{7/5})$. Hence from Lemma~\ref{lem:transitive-packing} the certificate size of $\mathcal{P}^*_H$ is $\Omega(\frac{n^3}{n^{7/5}}) = \Omega(n^{8/5})$.
Now from the Lemma~\ref{lem:certificate} we get $Q(f_H) =\Omega(n^{4/5})$.

\end{proof}

Hence from now onwards we assume that for all $i, |E(G_i)| = \Omega(n^{7/5})$. Moreover, we may also note that $\nu_H =O(n^{1/5}),$  if not we have a minimal 
certificate for $\mathcal{P}_H$ of size  $\Omega(n^{8/5})$. And hence from the Lemma~\ref{lem:certificate} we already get the desired bound of $Q(f_H) = \Omega(n^{4/5})$.

\begin{figure}[h]
\centering
\includegraphics[width=0.36\textwidth]{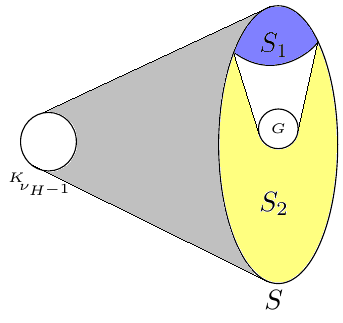}
\caption{The Restriction $\mathcal{P}''$: the hyper-clique $K_{\nu_H-1}$ is present, all the hyper-edges in the gray area are present. All the hyper-edges in blue region are present, all the hyper-edges in yellow region are absent. $G$ is fixed. White region symbolizes the hyper-edges with two-end points in $S_1$ and one in $S_2$ to be absent and one end point in $S_1$ and two end points in $S_2$ to be undertermined.}\label{fig:PH}
\end{figure}

Now we obtain a restriction $\mathcal{P}'$ of $\mathcal{P}_H$ as follows: divide $S$ into two parts say $S_1$ and $S_2$ of size $n_1$ and $n_2$ respectively, where we choose $n_1 = \Theta(n^{1/5})$ and $n_2 = \Theta(n)$. Set all the hyper-edges within $S_1$ to be present and set all the hyper-edges within $S_2$ to be absent. Also set all the hyper-edges with two endpoints in $S_1$ and one in $S_2$ to be absent. Only possible undetermined hyper-edges are with one endpoint in $S_1$ and two in $S_2$.
Note that even after setting all hyper edges in $S_1$ to be present we can safely assume that the property remains non-trivial. Otherwise we would have a certificate
for $\mathcal{P}_H$ of size $O(n^{3/5})$, hence the dual will have large ($ \gg \Omega(n^{8/5})$) certificates.

Let $G$ be a projection graph among all the $G_i$'s containing the least number of edges inside $S_2$. We further obtain a restriction $\mathcal{P}''$ by fixing a copy of $G$ inside $S_2$ and allowing only potential hyper-edges with one endpoint in $S_1$ and the other two endpoints forming an edge of $G$ (See Figure~\ref{fig:PH}).

Let $C$ be a vertex cover of $H$ of minimum cardinality. Note that in order to satisfy $\mathcal{P}_H$, at least one of the vertices from $C$ must move to $S$.
Let us call a vertex of $C$ that moves to $S$ as pivot. Let $k$ be the largest integer such that $P_H$ has a minimal certificate with $k$ pivots. 
Note that from Lemma~\ref{lem:pivot} each pivot has $\Omega(n^{7/5})$ edges incident on it.  Therefore if $k > n_1/2$ then we already have a minimal
certificate whose size is $\Omega(n^{8/5})$. Otherwise: $k \leq n_1 / 2$. First we argue that any pivot must belong to $S_1$. If on the contrary, it were in $S_2$
then the only possible edges incident on such a pivot $v$ are  of the form $(v, u, w)$ where $u \in S_1$ and $w \in S_2$.  But there can be at most $O(n^{6/5})$
such edges, which contradicts the fact that any pivot supports at least $\Omega(n^{7/5})$ edges.
Let the degree of a pivot be the number of edges inside $S_2$ that are adjacent to it.
Next we choose a certificate for $\mathcal{P}_H$ with at most $k \leq n_1/2$ pivots such that the degree of the minimum degree pivot is minimum possible.
Then we leave aside the minimum degree pivot in this certificate and  fix the $k-1$ other pivots and their projection on $S_2$.
From each of the remaining $n_1 - k + 1$ vertices we keep the projection
of the minimum degree pivot on $S_2$ as the only possible edges. 

 It is easy to see from minimality of our choice
 that at least one of these vertices must have all these $\Omega(n^{7/5})$ edges in order for the original graph to contain $H$.
 Thus we get an $\bigvee_{\Omega(n^{1/5})} \bigwedge_{\Omega(n^{7/5})}$ function as the restriction.

Since an $OR \circ AND$ on $m$ variables admits an $\Omega(\sqrt{m})$ lower bound on the quantum query complexity we get $Q(f_H) = \Omega(n^{4/5})$.

\noindent {\sf Case 2:} $\alpha_H \leq n/2$.
In this case we use Lemma~\ref{lem:3-uniform-threshold}. Since $n - \alpha_H \geq n/2$,
we get $Q(f_H) = \Omega(\sqrt{\alpha_H \cdot n})$.

Let $d$ denote the average degree of $H$. We consider two cases.

{\sf Case 2a:}   $d > n^{2/3}$.

In this case $|E(H)| > \Omega(n^{5/3})$. Hence from Lemma~\ref{lem:certificate} we get an $\Omega(n^{5/6})$ bound.

{\sf Case 2b:} $d \leq n^{2/3}$.

Here we use the extension of Tur\'an's Theorem (See Lemma~\ref{lem:extended-turan}
) to 3-uniform hypergraphs. Since the average degree is $O(n^{2/3})$, we get $\alpha_H \geq \Omega(n^{2/3})$. Therefore from Lemma~\ref{lem:3-uniform-threshold} we get $Q(f_H) = \Omega(n^{5/6})$.

This completes the proof of Theorem~\ref{thm:3-uniform}.

\hfill $\Box$

In the following two sections we study the Subgraph Homomorphism Problem. We first prove the quantum query complexity lower bounds for graphs and then for 3-uniform hypergraphs.

\section{Subgraph Homomorphism for Graphs}
\label{sec:homo}

\subsection*{Proof of Theorem~\ref{thm-homo}.}

Let $\chi(H)$ denote the chromatic number of $H$. Note that $H$ has a homomorphism into $K_t$ for $t = \chi(H)$, i.e., 
$f_{[H]}(K_{t-1}) = 0$ and $f_{[H]}(K_t) = 1$.

We consider the following two cases.

{\sf Case 1:} $t \geq n/2$:  

 In this case, it is easy to see that
 the minimum certificate size, $m(f_{[H]}) = \Omega(t^2) = \Omega(n^2).$ Hence from Lemma~\ref{lem:certificate} we get
 an $\Omega(n)$ lower bound on the quantum query complexity.

{\sf Case 2:} $t < n/2$: 

Consider the following restriction: We set a clique $K_{t-2}$ on $t-2$ vertices to be present and we also set all the edges from the remaining $n-t+2$ vertices to this clique to be present. Now notice that as soon as there is an edge between any two of the remaining  $n-t+2$ vertices, we have a $K_t$. Hence the property $f_{[H]}$ has become the property of containing an edge among the $n-t+2$ vertices. Since $t < n/2$, this is an OR function on $\Omega(n^2)$ variables. Thus $Q(f_{[H]}) = \Omega(n)$. 

\hfill $\Box$

\begin{remark} Our proof in fact shows that the minimum certificate size of
either $f_{[H]}$ or $f_{[H]}^*$ is $\Omega(n^2)$. Hence we also obtain
\[R(f_{[H]}) = \Omega(n^2)\].

\end{remark}

We now proceed to prove the quantum query complexity lower bound of the Subgraph Homomorphism Problem for 3-uniform hypergraph.

\section{Subgraph Homomorphism for $3$-Uniform Hypergraphs}
\label{sec:homo-appx}
\subsection*{Proof of Theorem~\ref{thm:homo-3uniform}.}

Proof of this theorem is similar to proof of Theorem~\ref{thm-homo}. 

Let $\chi(H)$ denote the chromatic number of $H$. Note that $H$ has a homomorphism into $K_t$ for $t = \chi(H)$, i.e., 
$f_{[H]}(K_{t-1}) = 0$ and $f_{[H]}(K_t) = 1$.

We consider the following two cases.

{\sf Case 1:} $t \geq n/2$:

 Unlike the graph homomorphism case, we cannot claim the presence of a $K_t$ in this case. However we can still use the following fact: 
 \begin{fact} (Alon \cite{alon85})
 If $H$ is a 3-uniform hypergraph which is not $k$ colorable then 
 \[|E(H)| = \Omega(k^3).\]
 \end{fact}
 Therefore, the minimum certificate size $m(f_{[H]}) = \Omega(t^3) = \Omega(n^3).$ Hence from Lemma~\ref{lem:certificate} we get an
 $\Omega(n^{3/2})$ lower bound on the quantum query complexity.

{\sf Case 2:} $t < n/2$: 

Consider the following restriction: We set a clique $K_{t-3}$ on $t-3$ vertices to be present and we also set all the edges from remaining $(n-t+3)$ vertices to this clique to be present. Now notice that as soon as there is an edge between any three of the remaining  $(n-t+3)$ vertices, we have a $K_t$. Hence the property $f_{[H]}$ has become the property of containing an edge among the $n-t+3$ vertices. Since $t < n/2$, this is an OR function on $\Omega(n^3)$ variables. Thus $Q(f_{[H]}) = \Omega(n^{3/2})$. 

\hfill $\Box$

\section{Conclusion \& Open Ends}
\label{sec:open}
We obtained an $\Omega(n^{3/4})$ lower bound for the quantum query complexity of Subgraph Isomorphism Problem for graphs, improving upon previously known $\Omega(n^{2/3})$ bound for the same. We extend our result to the $3$-uniform hypergraphs by exhibiting an $\Omega(n^{4/5})$ bound, which improves on previously known $\Omega(n^{3/4})$ bound. 
Besides the obvious question of settling the randomized and quantum query complexity of the Subgraph Isomorphism problem, there are a few interesting questions that might be approachable. We list some of them below:
\begin{question} 
Is it true that for any $n$-vertex graph $H$ we have: \begin{itemize}
\item [(a)] $R(f_H) = \Omega(\alpha_H \cdot n)$?

\item [(b)] $R(f_H) = \Omega(n^2/d_{avg}^H))$?

\item [(c)] $R(f_H) = \Omega(n^2/\chi_H)$?
\end{itemize}
\end{question}

\begin{question} Is it true that for any $3$-uniform hypergraph $H$ we have:
\[Q(f_H) = \Omega(n)?\]
\end{question}

\section{Acknowledgement}
The authors would like to thank the anonymous reviewers for their valuable comments and suggestions to improve this article. 
The authors also thank Biswaroop Maiti for the discussion in the beginning phase of this work.

\bibliographystyle{plain}
\bibliography{file.bbl}

\end{document}